\documentclass[reprint,aps,pre]{revtex4-1}

\usepackage{natbib}
\usepackage{amsmath,amssymb,amsfonts,amsthm}
\usepackage{graphicx}
\usepackage{hyperref}

\newtheorem{theorem}{Theorem}
\newtheorem{lemma}[theorem]{Lemma}
\newtheorem{corollary}[theorem]{Corollary}
\newtheorem{proposition}[theorem]{Proposition}
\theoremstyle{definition}
\newtheorem{definition}[theorem]{Definition}

\begin{document}
\title{Horizon Visibility Graphs and Time Series Merge Trees are Dual}
\author{Colin Stephen}
\affiliation{Coventry University, UK}
\email{colin.stephen@coventry.ac.uk}
\date{\today}

\begin{abstract}
In this paper we introduce the horizon visibility graph, a simple extension to the popular horizontal visibility graph representation of a time series, and show that it possesses a rigorous mathematical foundation in computational algebraic topology. This fills a longstanding gap in the literature on the horizontal visibility approach to nonlinear time series analysis which, despite a suite of successful applications across multiple domains, lacks a formal setting in which to prove general properties and develop natural extensions. The main finding is that horizon visibility graphs are dual to merge trees arising naturally over a filtered complex associated to a time series, while horizontal visibility graphs are weak duals of these trees. Immediate consequences include availability of tree-based reconstruction theorems, connections to results on the statistics of self-similar trees, and relations between visibility graphs and the emerging field of applied persistent homology.
\end{abstract}
\keywords{Horizontal Visibility Graph; Topological Data Analysis; Merge Tree; Time Series}
\maketitle

{\it Introduction.} The (directed) horizontal visibility graph \cite{Luque2009} or (D)HVG of a time series $\tau=(x_1,\ldots,x_N)$ is a network with nodes $\{1,\ldots,N\}$ and edges $(i,j)$ for each pair $i<j$ such that $i<k<j$ implies $x_k<x_i,x_j$. The undirected version omits the order of $i$ and $j$. Despite its structural simplicity, this graph captures much of the geometry of $\tau$ while remaining invariant under (positive) affine transformations.

Exact analysis and numerical simulation of HVGs shows that their properties, such as degree distributions and block entropies, bear an intimate relation to the dynamic properties of the system generating a time series. For example using the HVG one can determine whether an observed system is chaotic or stochastic and can estimate numerical values of key dynamic parameters including reversibility, Lyapunov exponents and Hurst indices \cite{Zou2019,Nunez2012}. This generality has led to successful applications ranging through cardiology, neurophysiology, meteorology, geophysics, protein dynamics and the financial markets \cite{Madl2016,Zhu2014,Schleussner2015,Braga2016,Zhou2014,Flanagan2016}. In many cases the statistics of HVG degree sequences and their subsequence motifs are the main discriminatory feature, and work is ongoing to fully understand why this feature is so effective from a theoretical context \cite{Gutin2011,Luque2017}. 

Topological data analysis (TDA) for time series follows a seemingly different path \cite{Edelsbrunner2010,Ghrist2008a,Perea2014,Mittal2017}. Beginning from a piecewise linear interpolation $\mathrm{PL}_\tau:\mathbb{R}\to\mathbb{R}$ of $\tau$, or from a distance or density function on a higher dimensional delay embedding of $\tau$, persistent homology tracks how the connected components of $\lambda$-sublevel sets $\{x : \mathrm{PL}_\tau(x) \leq \lambda \}$ merge as the threshold $\lambda$ increases over $\mathbb{R}$. The resulting merge tree is a rooted metric tree which has a natural branch decomposition structure summarised in a multiset of intervals called the barcode or persistence diagram of $\mathrm{PL}_\tau$. This multiset is the central object of study in persistent homology, and metrics on the space of barcodes and individual barcode statistics such as entropies are the main discriminatory features in applications of TDA to time series. They detect and quantify many of the same dynamical properties of a system generating a time series that are captured by HVGs \cite{Emrani2014a,Khasawneh2016,Khasawneh2018,Perea2015,Gidea2018,Mittal2017,Tempelman2019}.

{\it Contribution.} The wide ranging overlap between practical applications of TDA and HVGs to time series is not yet reflected in theory, but an intimate connection exists. It arises from a very simple shift in perspective: given a time series, instead of considering its merge tree with respect to a piecewise interpolation, we study its merge tree over a particular weighted graph. After making this change the branching of the tree exactly reflects the hierarchical nesting of the edges in a structure we call the \emph{horizon visibility graph}, which extends the standard HVG with two additional vertices representing the past and future. Establishing this duality involves fixing an appropriate embedding for the merge tree, then proceeding recursively on a subtree decomposition of that tree. We show that metric data on the tree imply that its subtrees correspond to recursively nested subgraphs of the horizon visibility graph. As a corollary HVGs are weak duals of merge trees, a connection which suggests several directions for developing the visibility approach.

\section*{Horizon Visibility Graph and Time Series Merge Tree Duality}
\label{sec:the_result}

Relevant concepts from topology are defined here in terms of graph theory. This simplifies the presentation and illustrates the connection to HVGs more clearly. For additional details and general topological definitions see \cite{Zomorodian2005,Edelsbrunner2010}. All time series, graphs and trees are finite. Without loss of generality time series are strictly positive. We begin with our simple extension to the HVG.

\begin{definition}\label{def:horizon_visibility_graph}
Given a time series $\tau=(x_1,\ldots,x_n)$ its \emph{horizon visibility graph} $\mathrm{HVG}_\infty(\tau)$ is defined to be the horizontal visibility graph of $\tau_\infty = (\infty,x_1,\ldots,x_n,\infty)$. 
\end{definition}

The remainder of this section provides a formal foundation for the graph $\mathrm{HVG}_\infty(\tau)$ in the framework of 0-dimensional homology over a filtered simplicial complex. Recall that a \emph{weighted graph} $G=(V,E,f)$ is a graph $(V,E)$ along with a weight function $f:E\to\mathbb{R}$. If the weight function is clear from context we use ``graph''. Assume all weights are positive.

\begin{definition}
Given a graph $G=(V,E,f)$ and $a\in\mathbb{R}$ define the \emph{sublevel graph} $G_a$ to be the subgraph of $G$ whose edges have weight no greater than $a$: $G_a := (V,E_a,f) \subseteq G$, where $E_a := \{ e\in E : f(e)\leq a \}$.
\end{definition}
Note that by definition all vertices of $G$ appear in its sublevel graphs and only edges are included or excluded depending on their weights.

\begin{lemma}\label{lem:sequence_of_sublevel_graphs}
The weight function $f$ on a finite graph $G=(V,E,f)$ induces a strictly increasing sequence of sublevel graphs of $G$ beginning at $(V,\emptyset)$ and ending at $G=(V,E)$.
\end{lemma}
\begin{proof}
Suppose that $a_1 < a_2 < \ldots < a_n$ are the distinct weights in the image $f(E)\subseteq \mathbb{R}$, write $G_i:=G_{a_i}$, and let $a_0 = 0$. Then every distinct threshold $a_i$ adds at least one new edge to $G_i$ that was not already in $G_{i-1}$ for $1\leq i\leq n$. Moreover since the $a_i$ exhaust all the distinct weights, all edges are included upon reaching the upper bound $a_n$. So we have a sequence $(V,\emptyset) = G_0 \subset G_1 \subset \ldots \subset G_n = (V,E)$ of strictly increasing sublevel graphs of $G$.
\end{proof}

\begin{definition}
Given an acyclic graph $G=(V,E,f)$ and $a\in\mathbb{R}$ say that two vertices $v,w\in V$ are \emph{$a$-connected} when any path in $G$ between them contains no weight exceeding $a$. Additionally say $v$ and $w$ are \emph{maximally $a$-connected} when any path in $G$ extending an $a$-connected path between them is not itself $a$-connected. 
\end{definition}

Being maximally $a$-connected is clearly an equivalence relation on $V$ for each $a\in\mathbb{R}$. For certain graphs the resulting one parameter family of relations has a tree structure:
\begin{lemma}\label{lem:refinement_of_partition}
For a connected graph $G=(V,E,f)$ the relation of being maximally $a$-connected induces a refinement of partitions of $V$. The refinement has the structure of a rooted tree called the \emph{merge tree} $T_G$ of $G$, with $V$ as the root, $\{ \{v\} : v\in V\}$ the leaves, and internal vertices being the maximally $a$-connected components of $G$ induced by its edge weights.
\end{lemma}
\begin{proof}
By Lemma \ref{lem:sequence_of_sublevel_graphs} there is a strictly increasing sequence of sublevel graphs $(V,\emptyset)=G_0 \subset G_1 \subset \ldots \subset G_n = (V,E)$ corresponding to the distinct weights $a_0=0$ and $a_1<\ldots<a_n$ in $f(E)$. When $G$ is connected then for all $a\geq a_n$ the only $a$-equivalence class is the full vertex set $V$, giving the root of $T_G$. On the other hand the $a_0$-equivalence classes are singletons each containing an element of $V$, giving the leaves.

Between these extremes, each $a_i$-equivalence class at level $G_i$ is fully contained in some $a_{i+1}$-equivalence class at level $G_{i+1}$ since each $a$-connected component of $G$ is automatically $b$-connected for all $b\geq a$, and we have $a_{i+1}>a_i$ for all $i$ by definition. Iterating $i$ over $0,1,\ldots,n-1$ gives the tree structure. 
\end{proof}

Thus each vertex in a merge tree has a height $a_i\in f(E)$ and each edge spans some half open interval $[a_i,a_j)$ for $0\leq i<j\leq n = |f(E)|$ corresponding to the heights of its incident vertices. In what follows a proper subtree is taken to include the (half) \emph{root edge} above its root vertex covering this interval. Call a subtree \emph{principal} when it contains the descendants of all of its vertices, and without loss of generality also assume the root edge of a full merge tree has a fixed finite length $r\in(0,\infty)$, say $r=1$. Note that merge trees contain only vertices of degree $d=1$ or $d\geq 3$ and in general may be non-binary trees.

\begin{corollary}\label{cor:subtrees_correspond_to_connected_components}
The map taking a principal subtree $\Lambda\subseteq T_G$, whose root vertex is at height $a_i$ and whose root edge spans $[a_i,a_j)$ for some $i<j$, to the subgraph $\gamma_\Lambda$ of $G$ whose vertices are the leaves of $\Lambda$, is a bijection from the set of principal subtrees to the set of all distinct maximally $a$-connected components of $G$ for $a\in\mathbb{R}$. In particular the component $\gamma_\Lambda$ is maximally $a$-connected exactly for $a\in [a_i,a_j)$ under this map.
\end{corollary}

To illustrate this bijection consider the merge tree $T_G$ in Figure \ref{fig:merge_tree_of_graph}. The line at $a=4.5$ intersects the root edges of three principal subtrees of $T_G$. The associated sublevel graph components are: the single leftmost vertex `born' at $G_0$ which `dies' at $G_6$, the single rightmost vertex born at $G_0$ which dies at $G_5$, and the six vertex chain born at $G_4$ which dies at $G_5$. Death is always by inclusion in to a larger connected component.
\begin{figure}[htb]
\centering
\includegraphics[width=0.45\textwidth]{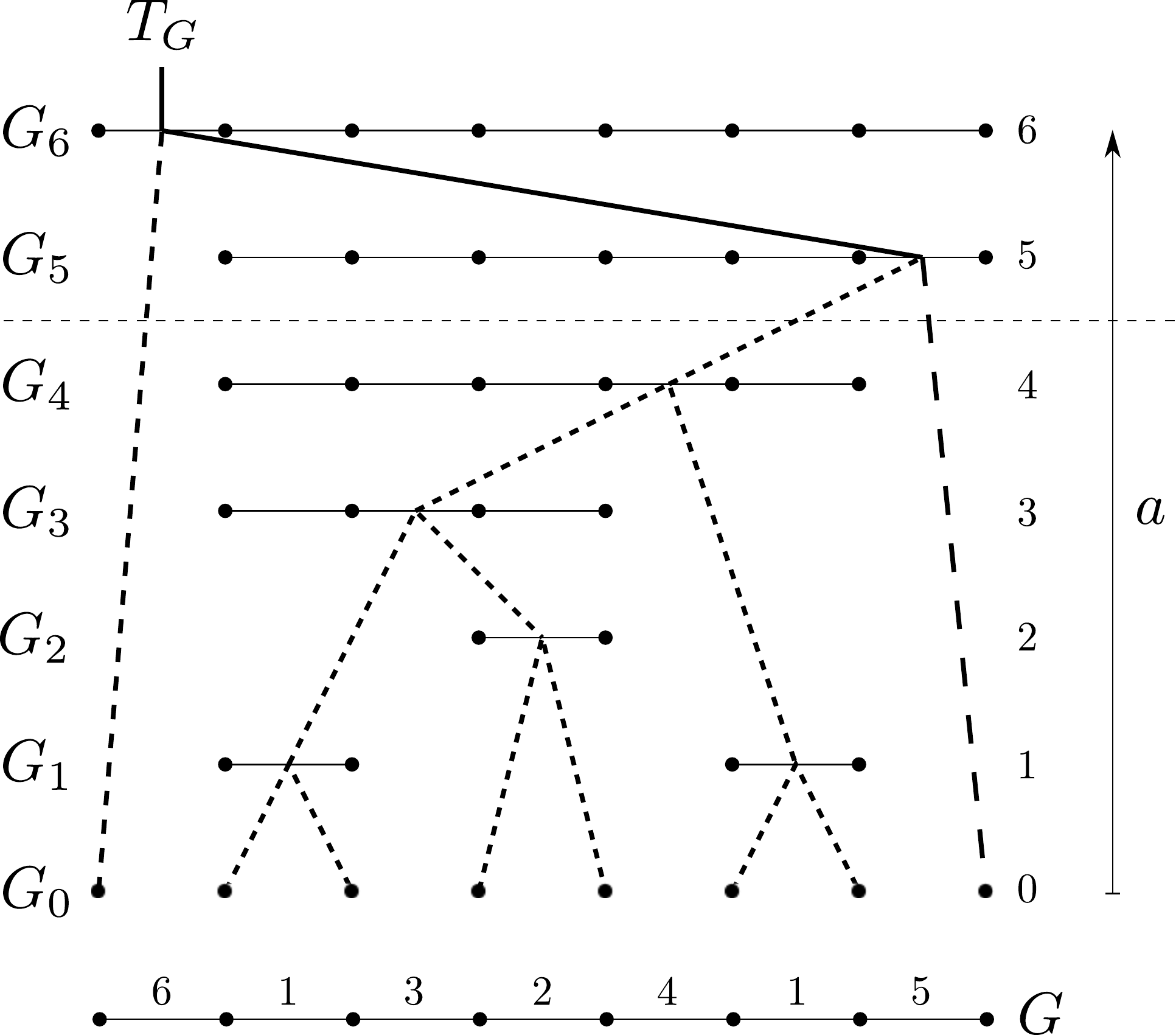}
\caption{Weighted graph $G$ and its merge tree $T_G$. Vertices of $T_G$ are connected components of the sublevel graphs $G_i$. Subtrees at level $a=4.5$ are described in the text.}\label{fig:merge_tree_of_graph}
\end{figure}
The bijection of Corollary \ref{cor:subtrees_correspond_to_connected_components} links each connected component appearing in some $G_i$ to a (principal) subtree rooted at it inside $T_G$ in the same way. In total there are fifteen subtrees of $T_G$ in the figure, including the eight single vertex trees, and there are also fifteen distinct $a$-connected components of $G$ as $a$ varies through $\mathbb{R}$, including the eight single vertices at $G_0$.

We are now ready to connect the idea of a merge tree over a graph to time series.
\begin{definition}
Given a time series $\tau = (x_1, \ldots, x_N)$ define the \emph{time series weighted path} $\check\tau$ to be the graph $\check\tau = (V,E,f)$ where:
\begin{itemize}
\item $V = \{0,1,\ldots,N\}$
\item $E = \{ e_i = (i-1,i) : 1\leq i \leq N\}$
\item $f:E\to \mathbb{R};e_i\mapsto x_i$ for $1\leq i\leq N$
\end{itemize}
If $\tau$ is empty then $\check\tau$ is defined as the graph with a single vertex and no edges.
\end{definition}
Since $\check\tau$ is connected and acyclic the following concept is well-defined.
\begin{definition}
The \emph{merge tree of a time series} $\tau$ is the merge tree $T_{\check\tau}$ of the weighted path $\check\tau$.
\end{definition}

Lemma \ref{lem:refinement_of_partition} implies that the leaves of $T_{\check\tau}$ are exactly the vertices $V=\{0,1,\ldots,N\}$ of $\check\tau$, so $T_{\check\tau}$ is actually an \emph{ordered tree}: we can order the children of any vertex according to the smallest leaves descended from them. This implies $T_{\check\tau}$ has an essentially unique plane embedding and so the discussion below is independent of the particular embedding chosen \cite{Gross2001}.

Consider $T_{\check\tau}$ embedded in the 2-sphere $S^2$ as follows.
Leaves are ordered anti-clockwise around the boundary $S^1$ of the disk $D^2$, and the vertex at the unbounded end of the root edge of $T_{\check\tau}$, labelled $\infty$, is placed on $S^1$ between leaves $0$ and $n$. Then points on $S^1$ are identified giving an embedding in $S^2$. Note that all leaves and the vertex $\infty$ are identified by this process. The initial $D^2$ embedding is shown grey in Figure \ref{fig:embedding_with_dual}.

The key outcome of this paper is that expressing the geometric relationships between connected regions inside $S^2 \setminus T_{\check\tau}$ with respect to the embedding above, and thus to any plane embedding, captures the geometric structure of $\tau$. Such relationships are described by the dual graph, but as $S^2$ is more difficult to visualise than $D^2$ we work in $D^2$ and adjust our definition of duality to compensate. Call points in $D^2$ \emph{external} when on $S^1$ otherwise \emph{internal}.
\begin{definition}\label{def:dual_of_embedding}
Given a graph $G$ embedded in $D^2$ define its \emph{dual} $G^*$ as follows. Internal vertices of $G^*$ are connected regions of $D^2\setminus (G\cup S^1)$ whose boundary does not include all of $S^1$. External vertices of $G^*$ are connected regions of $D^2\setminus (G\cup S^1)$ whose boundary \emph{does} include all of $S^1$, of which there is at most one. Edges in $G^*$ connect vertices whose primal regions share two sides of an edge, including with themselves.
\end{definition}

In particular the dual is well-defined for merge trees.
\begin{definition}\label{def:dual_of_merge_tree}
The \emph{dual} to a time series merge tree is the graph dual $T^*_{\check\tau}$ of its merge tree embedded in the closed disk $D^2$ as above. The \emph{metric dual} additionally copies edge lengths from the primal tree to the edges in the dual graph.
\end{definition}
An example of a dual to a time series merge tree is shown in Figure \ref{fig:embedding_with_dual}, illustrating the general relationship we now show: the dual is the horizon visibility graph.
\begin{figure}[htb]
\centering
\includegraphics[width=0.4\textwidth]{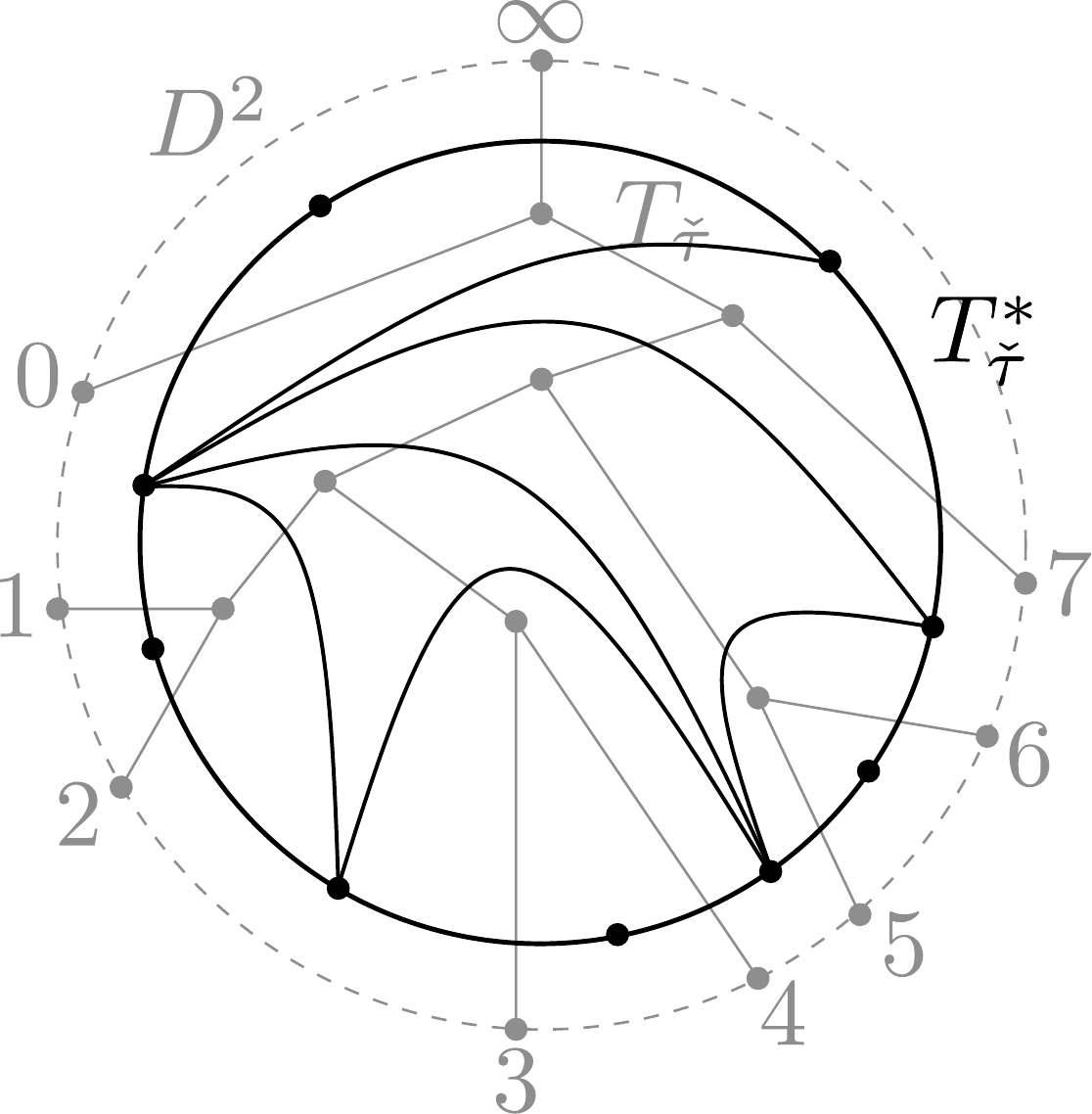}
\caption{A time series merge tree $T_{\check\tau}$ embedded in $D^2$ and its dual $T^*_{\check\tau}$. The latter is exactly $\mathrm{HVG}_\infty(\tau)$.}
\label{fig:embedding_with_dual}
\end{figure}

\begin{theorem}\label{thm:duality}
Given a time series $\tau = (x_1, \ldots, x_N)$ its horizon visibility graph $\mathrm{HVG}_\infty(\tau)$ is exactly the dual of its merge tree: $\mathrm{HVG}_\infty(\tau) = T^*_{\check\tau}$.
\end{theorem}

\begin{proof}
Let $\tau_\infty=(x_{-\infty},x_1,\ldots,x_n,x_\infty)$ where $x_{\pm\infty}=\infty$, and consider the embedding of $T_{\check\tau}$ in $D^2$ described above. Each value $x_i$ for $i=1,\ldots,n$ corresponds to the interterval on $S^1$ between leaves $i-1$ and $i$ of $T_{\check\tau}$. Similarly the values $x_{\pm\infty}$ correspond to the intervals on $S^1$ leading left and right from the root vertex labelled $\infty$. Thus the connected regions $\hat i$ in $D^2\setminus (T_{\check\tau}\cup S^1)$ for $i=-\infty,1,\ldots,n,\infty$ are in bijection with the values $x_i$ in $\tau_\infty$.

By the definition of duality for time series merge trees it then suffices to show that $x_i$ and $x_j$ are horizontally visible in $\tau_\infty$, written $x_i\sim x_j$, if and only if regions $\hat i$ and $\hat j$ in $D^2\setminus (T_{\check\tau}\cup S^1)$ share a unique boundary edge in $T_{\check\tau}$. In other words we want to show that $x_i\sim x_j \Longleftrightarrow |\hat i \cap_E \hat j| = 1$ where $\cap_E$ represents intersections of region boundaries, namely along edges. Since each pair of regions bounded by the tree and $S^1$ share at most one edge, it suffices to show that $\hat i \cap_E \hat j \neq \emptyset$.

Suppose $i<j$ and $x_i \sim x_j$. Then for all $k$ satisfying $i<k<j$ we know that $x_k < x_i,x_j$. Let $a^* := \max \{ x_k : i<k<j \}$ and $a_* := \min\{x_i,x_j\}$. Then for any $a\in [a^*,a_*)$ no edge on the path from vertex $i$ to vertex $j-1$ in $\check\tau$ exceeds $a$, but the edges $e_i$ and $e_j$ both do. In other words the pair $(i,j-1)$ is maximally $a$-connected in $\check\tau$ over this half open interval. By Corollary \ref{cor:subtrees_correspond_to_connected_components} there exists a principal subtree $\Lambda \subset T_{\check\tau}$, whose root edge $e_\Lambda$ corresponds to a maximally $a$-connected component of $\check\tau$ for $a\in [a^*,a_*)$ and whose leaf set is $L_\Lambda = \{i,i+1,\ldots,j-2,j-1\} \subset V$. Since $\Lambda$ is principal it has no other leaves and $e_\Lambda$ must be shared between regions $\hat i$ and $\hat j$ as illustrated in Figure \ref{fig:subtree_separation_of_regions}, so $\hat i \cap_E \hat j \neq \emptyset$ as required.
\begin{figure}[htb]
\centering
\includegraphics[width=0.45\textwidth]{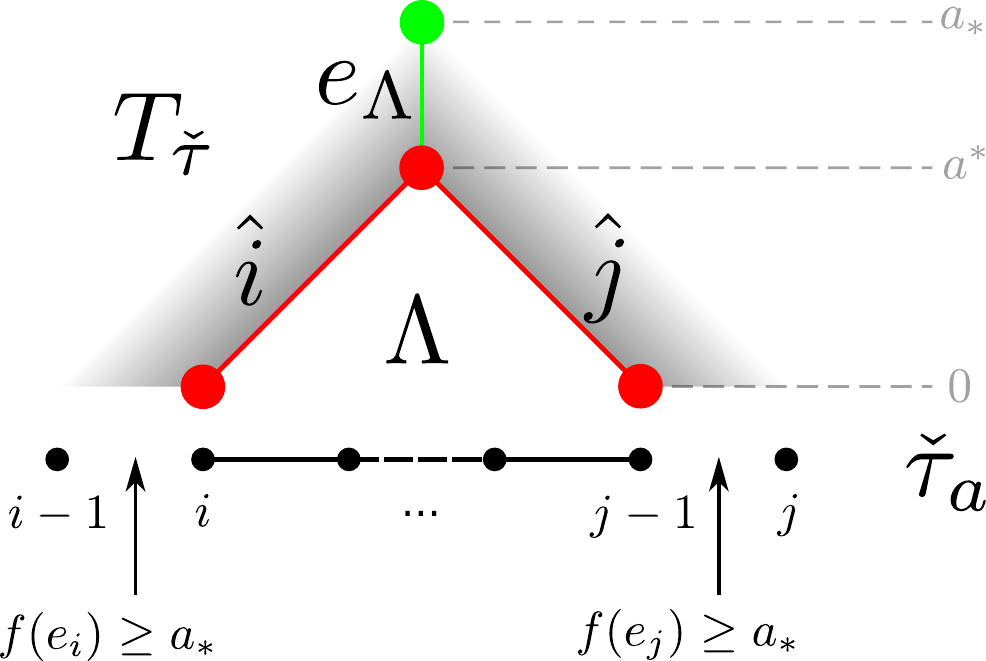}
\caption{The principal subtree $\Lambda\subset T_{\check\tau}$ spanning vertices $\{i,i+1,\ldots,j-1\}$ corresponds to a maximally connected component in $\check\tau$ when $x_i$ and $x_j$ are horizontally visible.}
\label{fig:subtree_separation_of_regions}
\end{figure}

In the other direction suppose that $i<j$ and regions $\hat i, \hat j$ share an edge $\hat e_{i,j} = \hat i \cap_E \hat j$ in $T_{\check\tau}$. Note that every ordered tree can be recursively decomposed into a fan of nonempty principal subtrees whose roots are the immediate children of the containing tree's root \cite{Dershowitz1980}. Since $\hat i$ and $\hat j$ share an edge this decomposition implies that there exists a principal subtree $\Lambda\subset T_{\check\tau}$ whose leaves are exactly $L_\Lambda = \{i,\ldots,j-1\}$ between the two regions. By Corollary \ref{cor:subtrees_correspond_to_connected_components} we are back in the situation illustrated in Figure \ref{fig:subtree_separation_of_regions}: $\Lambda$ corresponds to a connected component of $\check\tau$ that is maximally $a$-connected for $a$ in an interval $[\alpha,\beta)$ where $\alpha$ is the value at the root vertex of $\Lambda$ and $\beta$ is the lowest upper bound of values on the edge emerging upwards from the root. But by the construction of the merge tree we must have $\alpha = a^* = \max\{x_k:i<k<j\}$ being the maximum weight on the path between vertices $i$ and $j-1$, and $\beta = a_* = \min\{x_i,x_j\}$ the value at which the first neighbouring edge is added to the connected component spanning $L_\Lambda$ as $a$ increases. So $x_i\sim x_j$ as required.
\end{proof}

Time series merge trees are metric trees so Theorem \ref{thm:duality} allows us to extend Definition \ref{def:horizon_visibility_graph} to the following.
\begin{definition}
The \emph{persistence weighted} horizon visibility graph of a time series is the metric dual of its merge tree. In particular it has weights $p=\beta-\alpha$ on its edges where the half open interval $[\alpha,\beta)$ is spanned by the corresponding edge in the merge tree.
\end{definition}

Every rooted tree is naturally directed with all edges oriented towards, or away from, the root. Therefore Theorem \ref{thm:duality} also holds for directed horizon visibility graphs when a consistent rule for orienting dual edges is applied throughout the proof above. Moreover, since horizon visibility graphs are also horizontal visibility graphs with maximal endpoints we have a similar but weaker result for HVGs as follows.

\begin{definition}\label{def:weak_dual_of_merge_tree}
The \emph{weak dual} $T^\circ_{\check\tau}$ to a time series merge tree $T_{\check\tau}$ is the subgraph of its dual $T^*_{\check\tau}$ created by removing vertices $\widehat{\pm\infty}$.
\end{definition}

The term `weak' is used here because excluding connected regions $\widehat{\pm\infty}$ respects the intuition that regions whose boundary includes the infinite root edge are themselves unbounded, and such regions are omitted from the standard weak dual. With this intuition formalised the next result immediately follows.

\begin{corollary}\label{cor:weak_duality}
Given a time series $\tau=(x_1,\ldots,x_n)$ its horizontal visibility graph $\mathrm{HVG}(\tau)$ is exactly the weak dual of its merge tree: $\mathrm{HVG}(\tau)=T^\circ_{\check\tau}$.
\end{corollary}

Note however that in general the graph $\mathrm{HVG}(\tau)$ is \emph{not} dual to a tree.

\section*{Discussion}
\label{sec:horizon_visibility_graphs}

{\it Reconstruction Results.} The first thing Theorem \ref{thm:duality} and Corollary \ref{cor:weak_duality} imply is that the \emph{nesting structure} of edges in a visibility graph carry all of the relevant geometric information brought over from a time series. This helps explain the widely observed discrimination power of the degree sequence of an HVG, which by duality is the sequence of counts of internal boundary edges of regions under the merge tree. These strongly constrain the possible subtree decompositions a given tree can present. For example the next result follows quickly, where a canonical time series is one that is in general position except its end values are global maxima.
\begin{corollary}[\cite{Luque2017}]\label{cor:canonical_degrees}
Canonical horizontal visibility graphs are uniquely determined by their degree sequences.
\end{corollary}
\begin{proof}
The largest proper principal subtree of the merge tree of a canonical time series is a binary rooted tree. Ordered binary rooted trees with equal edge lengths and $n+1$ leaves $0,1,\ldots,n$ are uniquely determined by the $n$ leaf-to-leaf distances between order neighbours: $d_{0,1},d_{1,2},\ldots,d_{n-1,n}$. This follows by an induction on the number of leaves and observing that for some $i\in\{1,\ldots,n\}$ we have $d_{i-1,i}=2$, meaning a pair of edges can always be removed to give a strictly smaller tree with known distances between its remaining neighbours.  
\end{proof}
Indeed any constraint on a time series forcing the largest proper principal subtree of its merge tree to be binary implies the unique reconstruction of its HVG from its degree sequence in the same way.

Unique reconstruction results over arbitrary time series are less straightforward than over subclasses like the one in Corollary \ref{cor:canonical_degrees}. However there exist several theorems and algorithms for both binary and non-binary tree reconstruction developed over many decades for applications to phylogenetic trees and more widely \cite{Gusfield1997}. To take one example consider the well known \emph{neighbour joining} algorithm which uniquely reconstructs a tree from is full pairwise leaf-to-leaf additive distance matrix \cite{Saitou1987}. We can use this to quickly prove the following result.
\begin{theorem}
The in and out degree sequences of a directed horizon visibility graph uniquely determine it.
\end{theorem}
\begin{proof}
Suppose we are given in and out sequences $d^+ = (d_0^+,d_1^+,\ldots,d_n^+,d_{n+1}^+)$ and $d^- = (d_0^-,d_1^-,\ldots,d_n^-,d_{n+1}^-)$ for a horizon visibility graph $G$, where $d_0^\pm$ and $d_{n+1}^\pm$ are the in and out degrees of the regions to the left and right of the root vertex (labelled $n+1$ here) respectively. Consider its dual tree $G^*$ in $D^2$ and its $n+2$ connected regions. The degrees $d^+$ and $d^-$ fix the number of boundary edges each region has on its left and right hand sides, with respect to the root of the smallest principal subtree of $G^*$ containing the region. Using these values we can reconstruct the leaf-to-leaf distance matrix $D$ for the merge tree $G^*$ as follows. Due to the circular order of external vertices take all leaf labels and subscripts to be $\mathrm{mod}(n+2)$ from now on. 

Write $d_{i,j}$ for the path length from leaf $i$ to leaf $j$. Then we must have that $d_{i,i+1}=d_i^+ + d_i^-$ for $0\leq i\leq n+1$ giving the first off-diagonal in $D$. Writing $d_{i,i+1}^+$ for $d_i^+$ and $d_{i,i+1}^-$ for $d_i^-$ it is straightforward to show that we get a recurrence for the second off-diagonal: $d_{i,i+2} = d_{i,i+2}^+ + d_{i,i+2}^-$ for $0\leq i\leq n$, where the summands are given by $$d_{i,i+2}^\pm = d_{i,i+1}^\pm + d_{i+1,i+2}^\pm - \min(d_{i,i+1}^-, d_{i+1,i+2}^+).$$
Similar recurrences give the remaining off diagonals. Once this is done we can apply the neighbour joining algorithm to $D$ to reconstruct the correct unrooted merge tree topology. Since the root edge is already known via its exterior vertex label $n+1$ we have the correct rooted tree as well. Finally compute the oriented dual to give the horizon visibility graph. 
\end{proof}

{\it Trend Detection.} An interesting practical impact of moving to horizon visibility, beyond simpler reconstruction results, is that key geometric information about leading and trailing trends in the data is no longer lost. Horizontal visibility graphs are unable to distinguish between simple trends such as $\tau_1 = (1,2,3)$, $\tau_2 = (3,2,1)$, $\tau_3 = (1,2,1)$, and $\tau_4 = (1,1,1)$ because the weak dual omits edges on the outer boundary of a merge tree. However as shown in Figure \ref{fig:different_duals} horizon visibility graphs detect the difference. This implies sliding window techniques over data with statistical trends will be able to detect the trends using, for example, expected degree sequences.
\begin{figure}[htb]
\centering
\includegraphics[width=0.45\textwidth]{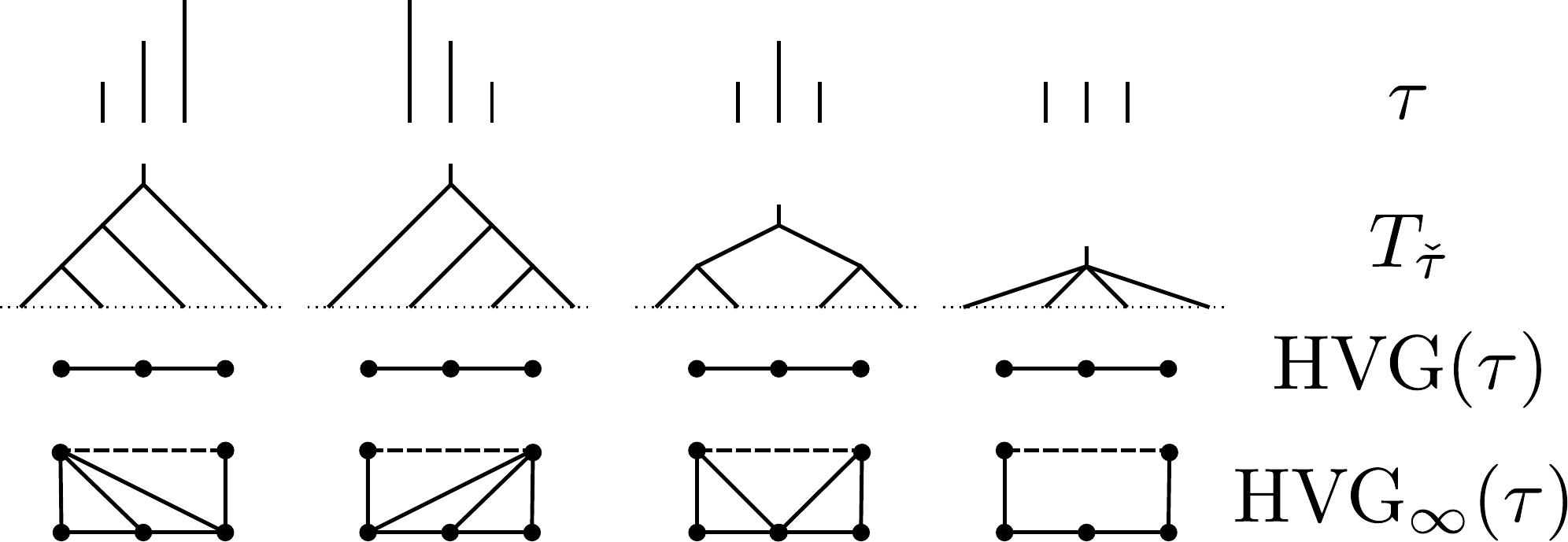}
\caption{Merge trees and their horizontal and horizon visibility graphs for simple trends. The root edge is dashed in $\mathrm{HVG}_\infty$.}
\label{fig:different_duals}
\end{figure} 

{\it Connections to TDA.} There is a direct connection between the time series merge trees defined in this paper and the merge trees underlying earlier applications in TDA, which are typically over continuous domains. The connection is captured by the following propositions, which are straightforward to prove. A single \emph{Horton pruning} of a tree is the operation of cutting off its leaves and their parental edges from the tree, then removing any remaining chains of degree-two vertices. Such operations have been studied in the context of quantifying the fractal dimension of random trees \cite{Zaliapin2012,Kovchegov2016}.
\begin{proposition}
The branch structure of the merge tree of a piecewise linear interpolation of a finite time series is exactly the first Horton pruning of its horizon visibility graph (via duality).
\end{proposition}
Similarly, when metric data are included we can recover the barcodes studied by TDA in full. For details of how the \emph{Elder Rule} computes barcodes on trees see \cite{Curry2018}.
\begin{proposition}
The Elder Rule on the first Horton pruning of a persistence weighted horizon visibility graph gives the barcode of its piecewise linear interpolation.
\end{proposition}
This makes horizon graphs more sensitive to certain features. For example the horizon visibility graph detects \emph{monotonic subsequences} in a time series, which are invisible to trees over piecewise linear and piecewise constant interpolations, as shown in Figure \ref{fig:monotonic_subsequences}. This means they can detect changes in frequency more readily, which could be useful for topologically aware signal analysis. Similarly, metric data on the persistence weighted graph quantify the \emph{scales} at which different geometric features exist. So their weighted degree sequences, extending the combinatorial degree sequence, can distinguish between features with the same geometry appearing at different scales.
\begin{figure}[htb]
\centering
\includegraphics[width=0.4\textwidth]{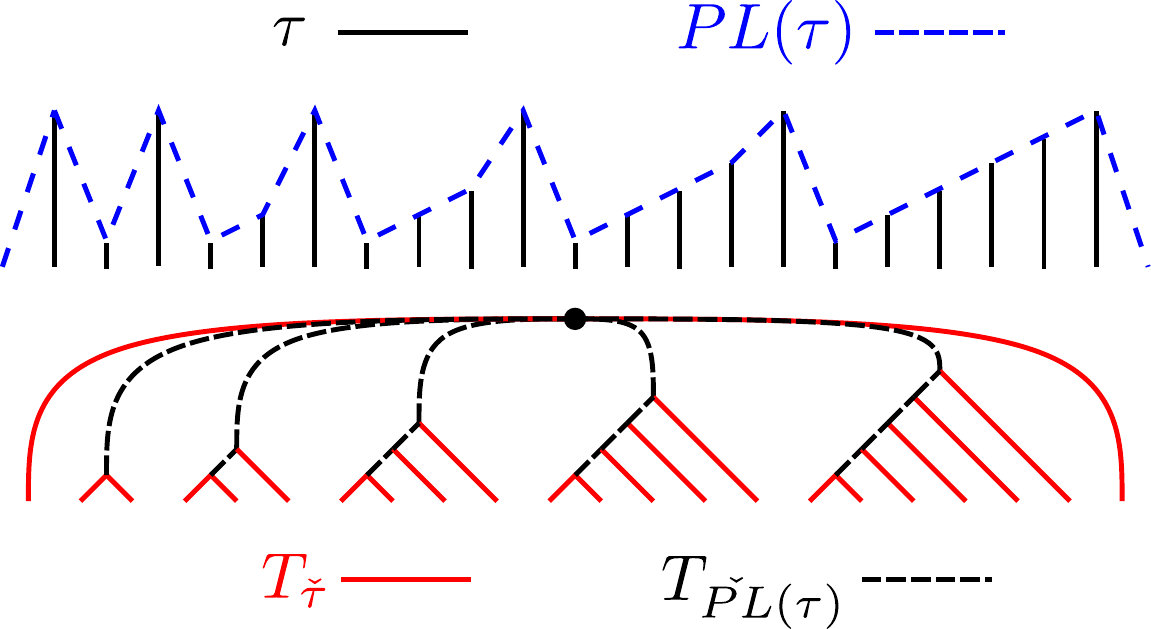}
\caption{A time series $\tau$ with monotonic subsequences of different lengths and its piecewise linear interpolation (top). The time series merge tree $T_{\check\tau}$ and its first Horton pruning (bottom). The pruned tree is the merge tree of $PL(\tau)$.}
\label{fig:monotonic_subsequences}
\end{figure}

{\it Conclusion.} Horizon visibility graphs simultaneously extend and unify HVGs and topological merge trees over piecewise interpolations of sequences. In doing so they add the ability to detect trends and the scale of geometric features, absent from HVGs, and the ability to detect monotonic subsequences and thus frequency-based geometric features, absent from trees over piecewise linear interpolations. More importantly there exists a wide body of work on the theoretical properties of combinatorial and metric trees in general \cite{Drmota2009,Evans2006} and topological merge trees applied to data analysis in particular \cite{Beketayev2014}, that applies to these graphs. This setting offers a number of directions to build on the connections established here.

Finally, it should be noted that while horizon visibility graphs are dual to trees, their graph properties capture features that may not be apparent in the tree representation. In particular the tree analogue of degree sequences, leaf-to-leaf path lengths between ordered neighbours, is not widely studied in applications of trees, so the horizon and horizontal visibility representations continue to express useful and complementary features of their own.


\end{document}